\documentclass[11pt]{article}

\usepackage{algorithm}
\usepackage[noend]{algpseudocode}

\usepackage{hyperref}
\hypersetup{
    unicode=false,          
    colorlinks=true,        
    linkcolor=darkblue,          
    citecolor=darkgreen,        
    filecolor=magenta,      
    urlcolor=cyan           
}
\usepackage{xcolor}
\definecolor{darkgreen}{rgb}{0,0.5,0}
\definecolor{darkblue}{rgb}{0,0,0.8}

\usepackage{tikz}
\usepackage{latexsym}
\usepackage{times}
\usepackage{amsthm}
\usepackage{amsfonts}
\usepackage{amssymb}
\usepackage{paralist}
\usepackage{amsmath}
\usepackage[letterpaper,margin=2.3cm]{geometry}
\usepackage{appendix}
\usepackage{graphicx}
\usepackage{enumerate}
\usepackage{tabularx}
\usepackage[compact]{titlesec}
\usepackage[]{subfig}

\newcommand{\set}[1]{\left\{#1\right\}}
\newcommand{\floor}[1]{\left\lfloor #1 \right\rfloor}
\newcommand{\ceil}[1]{\left\lceil #1 \right\rceil}

\newcommand{\hide}[1]{}

\renewcommand{\include}{\input}

\newcommand{\todo}[1]{{\color{red}\bf [TODO: #1]}}
\renewcommand{\paragraph}[1]{\smallskip\noindent\textbf{#1: }}

\newtheorem{theorem}{Theorem}[section]
\newtheorem{lemma}[theorem]{Lemma}

\newtheorem{corollary}[theorem]{Corollary}
\newtheorem{definition}{Definition}[section]

\newcommand{\calM}{\ensuremath{\mathcal{M}} }
\newcommand{\RN}[1]{%
  \textup{\uppercase\expandafter{\romannumeral#1}}%
}

\title{Multi-Message Broadcast in Dynamic Radio Networks\footnote{Research supported by ERC Grant No. 336495 (ACDC)}}
\author{Mohamad Ahmadi, Fabian Kuhn
\\[.4cm]
Department of Computer Science\\
University of Freiburg\\
Freiburg, Germany\\
{\small \texttt{\{mahmadi, kuhn\}@cs.uni-freiburg.de}}}

\date{}

\begin{document}

\maketitle              

\thispagestyle{empty}

\begin{abstract}
   We continue the recent line of research studying information
  dissemination problems in adversarial dynamic radio networks. We
  give two generic algorithms which allow to transform generalized
  version of single-message broadcast algorithms into multi-message
  broadcast algorithms. Based on these generic algorithms, we obtain
  multi-message broadcast algorithms for dynamic radio networks for a
  number of different dynamic network settings. For one of the
  modeling assumptions, our algorithms are complemented by a lower
  bound which shows that the upper bound is close to optimal.
\end{abstract}
\vspace{.5cm}

\section{Introduction}
\label{sec:intro}
When developing algorithms for wireless networks, one has to deal with
unique challenges which are not or much less present in the context of
wired networks. All nodes share a single communication medium and
whether a transmitted signal can be received by a given node might
therefore depend on the behavior of all other nodes in the
network. Moreover, the reception of a wireless signal can be influenced by additional wireless devices, multi-path
propagation effects, other electric appliances, or also various
additional environmental conditions, see e.g.,
\cite{kim06,newport:2007,ramachandran07,yarvis02,srinivasan08}. As a
consequence, wireless connections often tend to be unstable and
unreliable. Moreover, wireless devices might be mobile, in which case
connectivity changes even when ignoring all the above effects. We
therefore believe that in order to study such dynamic and
unpredictable communication behavior, it is important to also study
unreliable and dynamic variants of classic wireless communication
models such as, e.g., the classic radio network model introduced in
\cite{chlamtac85,bgi1} or the more complex but also more realistic
models like the SINR model \cite{gupta:2000,moscibroda06} or the
affectance model \cite{affectance}.

In recent years, there has been a considerable effort in investigating
basic communication and distributed computation problems in radio
network models which exhibit adversarial dynamic, nondeterministic
behavior. In \cite{clementi04}, Clementi et al.\ study the problem of
broadcasting a single message in a synchronous dynamic radio network
where in each round a subset of the links might fail
adversarially. Communication is modeled using the standard graph-based
radio network model \cite{bgi1}. In each round, a node can either
transmit or listen and a node $u$ receives a message transmitted by a
neighbor $v$ in a given round $r$ if and only if $v$ is the only
round-$r$ neighbor of $u$ transmitting in round $r$. The paper studies
deterministic algorithms and it in particular shows that if $D$ is the
diameter of the fault-free part of the network, deterministic
single-message broadcast requires time $\Theta(Dn)$, where $n$ is the
number of nodes of the network. In \cite{clementi09}, Clementi et al.\
study an even more dynamic network model where the network topology
can completely change from round to round. It is in particular shown
that the single-message broadcast problem can be solved in time
$O(n^2/\log n)$ by a simple randomized algorithm where in each round, each node
knowing the broadcast message, transmits it with probability
$\ln(n) / n$. It is also shown that the asymptotic time complexity of
this algorithm is optimal. A similar model to the one in
\cite{clementi04} has been studied in \cite{dualgraph}. In the
\emph{dual graph} model of \cite{dualgraph}, it is assumed that the set
of nodes is fixed and the set of edges consists of two parts, a
connected set of reliable edges and a set of unreliable edges. In each
round, the communication graph consists of all reliable and an
arbitrary (adversarially chosen) subset of the unreliable edges. Among
other results, it is shown that there is a randomized algorithm which
achieves single-message broadcast in time $O(n\log^2 n)$. The
algorithm in \cite{dualgraph} works in the presence of a
\emph{strongly adaptive} adversary which determines the set of edges
of a given round $r$ after the distributed algorithm decides what
messages are transmitted by which nodes in round $r$. In
\cite{obliviousDG}, the same problem is considered for weaker
adversaries. A \emph{weakly adaptive} adversary has to determine the
topology of round $r$ before a randomized distributed algorithm
determines the randomness of round $r$ and an \emph{oblivious
  adversary} has to determine the whole sequence of network topologies
at the very beginning of the execution of a distributed
algorithm. Additional problems in the dual graph model of
\cite{dualgraph,obliviousDG} have also been studied in
\cite{DG_structuring,ghaffari16,DG_localbroadcast}.

The dynamic network models of the above papers can (mostly) be seen as
the extreme cases of the \emph{$T$-interval connected} dynamic graph
model of \cite{kuhn:stoc10}. For a positive integer $T$, a dynamic
network is called $T$-interval connected, if for any time interval of
$T$ consecutive rounds the graph which is induced by the set of edges
present throughout the whole time interval is connected. Hence, for
$T=1$, the network graph has to be connected in every round, whereas
for $T=\infty$, we obtain the dual graph model. In \cite{opodis15},
the single-message broadcast problem has been studied for general
$T$-interval connected dynamic radio networks and for a more
fine-grained adversary definition. More specifically, for an integer
$\tau\geq 0$, an adversary is called $\tau$-oblivious if the network
topology of round $r$ is determined based on the knowledge of all
random decisions of the distributed algorithm of rounds
$1,\dots,r-\tau$. Hence, an oblivious adversary is $\infty$-oblivious,
a weakly adaptive adversary is $1$-oblivious, and a strongly adaptive
adversary is $0$-oblivious.
\\[.1cm]
\noindent
\textbf{Additional Related Work.} 
Information dissemination and specifically broadcasting is a
fundamental problem in networks and therefore there exists a rich
literature on distributed algorithm to solve broadcast in various
communication settings. In particular, the problem is 
well-studied in static wireless networks (see, e.g.,
\cite{bgi1,bgi2,chlamtac85,chlebus:2000,clementi:2001} and many
more). More recently, the multi-message problem has been studied
quite extensively in the context of wireless networks, see e.g.,
\cite{chrobak:2004,Chlebus09,Christersson,gasieniec,khabbazian11,maclayer,kuhn:2011}.

We also note that several dynamic network models similar to the one in
\cite{kuhn:stoc10} have been studied prior to \cite{kuhn:stoc10}, for
example in
\cite{avin08,baumann09,time-varying,clementi09,odell05}. For a recent
survey, we also refer to \cite{kuhn:2011:survey}. In addition to the
works already mentioned, information dissemination and other problems
in faulty and dynamic radio network models have also been studied in,
e.g., \cite{fernandezanta12,clementi04,kranakis01}.

\subsection{Contributions}
\label{sec:results}

In the present work, we extend \cite{opodis15} and more generally the
above line of research and we study the multi-message broadcast
problem in $T$-interval connected dynamic radio network against a
$\tau$-oblivious adversary. For some $s\geq 1$, we assume that there
are $s$ broadcast messages, each initially given to a single source
node. A randomized distributed algorithm solves the $s$-multi-message
broadcast problem if with high probability (w.h.p.), it disseminates
all $s$ broadcast messages to all nodes of the network. We say that
the communication capacity of a network is $c\geq 1$ if every message
sent by a node can contain up to $c$ different broadcast
messages.
	 
\subsubsection{Upper Bounds}

Most of our upper bounds are achieved by \emph{store-and-forward}
algorithms, i.e., by algorithms which treat the broadcast messages as
single units of data that cannot be altered or split into several
pieces. All our store-and-forward protocols are based on two generic
algorithms which allow to turn more general variants of single-message
broadcast algorithms into multi-message broadcast algorithms. In
addition, Theorem~\ref{thm:04} proves an upper bound that can be achieved
using random linear network coding. When dealing with linear network
coding algorithms, we use the common convention and disregard any
overhead created by the message headers describing the linear
combination of broadcast messages sent in the message (see, e.g.,
\cite{gossip}).\footnote{Note that this assumption is reasonable as long as
  the number of broadcast messages which are combined with each other
  is at most the length of a single broadcast message (in bits).}

\begin{theorem}\label{thm:01}
  Assume that we are given an $\infty$-interval connected dynamic
  $n$-node network controlled by a $0$-oblivious adversary.  Using
  store-and-forward algorithms, for communication capacity $c=1$,
  $s$-multi-message broadcast can be solved, w.h.p., in time
  $O\big(ns\log^2n\big)$, whereas for arbitrary $c\geq 1$, it can be
  solved, w.h.p., in time $O\big(\big(1+\frac{s}{c}\big)n\log^4 n\big)$.
\end{theorem}

\begin{theorem}\label{thm:02} 
  Consider the $s$-multi-message broadcast problem in $1$-interval
  connected dynamic $n$-node networks controlled by a $1$-oblivious
  adversary. Using store-and-forward algorithms, for communication
  capacity $c=1$, the problem can be solved, w.h.p., in time
  $O\big( (1+\frac{s}{\log n})n^2\big)$, and for arbitrary $c\geq 1$,
  it can be solved, w.h.p., in time $O\big((1+\frac{s}{c})n^2\log n\big)$.
\end{theorem}

\begin{theorem}\label{thm:03}
  Let $T \geq 1$ and $\tau \geq 1$ be positive integer parameters. In
  $T$-interval connected dynamic networks controlled by a
  $\tau$-oblivious adversary, for communication capacity $c=1$, the
  $s$-multi-message broadcast problem can be solved, w.h.p., in time
  \[
  O\left(\left(1+\frac{n}{min\{\tau , T\}}\right)\cdot ( s+\log
    n)\cdot n\cdot \log^3n\right),
  \]	
  and for an arbitrary $c$, it can be solved, w.h.p., in time
  \[
  O\left(\left(1+\frac{n}{min\{\tau , T\}}\right)\cdot \frac{ns}{c}\cdot\log^4n\right).
  \]
\end{theorem}

\begin{theorem} \label{thm:04} 
  Using linear network coding, in $1$-interval connected dynamic
  networks with communication capacity $1$ and a
  $1$-oblivious adversary, $s$-multi-message broadcast can be solved
  in time $O(n^2+ ns)$, w.h.p.
\end{theorem}

\subsubsection{Lower Bound}

\begin{theorem}\label{thm:05}
  In $\infty$-interval connected dynamic networks with communication
  capacity $c\geq 1$ and being controled by a $0$-oblivious adversary, any
  $s$-multi-message broadcast algorithm requires at least time
  $\Omega(ns/c)$, even when using network coding.  Further, there is a
  constant-diameter $\infty$-interval connected network with
  communication capacity $1$ such that any store-and-forward algorithm
  requires at least $\Omega\big((ns-s^2)/c\big)$ rounds to solve $s$-multi-message
  broadcast against a $0$-oblivious adversary.
\end{theorem}

\hide{
\subsection{Additional Related Work}
\label{sec:relwork}

\todo{Make this much shorter and more concise! Important references:
  additional work on dynamic radio networks, something general on
  broadcasting in radio networks, short discussion of other papers
  with similar dynamic network models.}
  
\noindent
	Information dissemination (spreading) or message broadcasting is one of the fundamental problems in distributed computing and therefore a rich literature exists on this topic.
	This problem, which is well-studied in static networks (e.g. \cite{bgi1,bgi2,chlamtac85,jurdzinski14} and many more), was recently considered in various dynamic network models.
	Information dissemination is widely studied in $T$-interval connected graph model (see e.g. \cite{kuhn:stoc10,opodis15,obliviousDG,dualgraph}).
	However, it is one of the main problems which is considered in many other dynamic network models like~\cite{fernandezanta12,baumann09,odell05}. 
	
	Similarly to the single-message broadcast problem, multi-message broadcast problem is also a well-studied problem in static networks~\cite{attiya,leighton,lynch,peleg}.

\vspace{2cm}
	Information dissemination (spreading) or message broadcasting is one of the fundamental problems in distributed computing and therefore a rich literature exists on this topic in various models. 
	Here we would like to consider three different problems regarding information dissemination in dynamic networks and mention part of the work has been done so far. 
\\[.1cm]
\noindent
\textit{Single-Message Broadcast.}
	This is a problem of broadcasting a single message from a source node to all nodes of a network. 		This problem is one of the well studied problems in static networks (e.g. \cite{bgi1,bgi2,chlamtac85,jurdzinski14} and many more). 
	 It is moreover studied in various models of dynamic networks, e.g., \cite{}. 
	 Let us consider the $T$-interval connectivity model against $\tau$-oblivious adversary introduced in~\cite{kuhn:stoc10}. For the extreme case of $T=1$, $\tau = 1$, a very simple and optimal broadcast protocol is introduced by Clementi et al.~\cite{clementi09}, where all the informed nodes independently transmit the source message with probability $\ln n/n$ in each round. 
	 For the other extreme case of $T=\infty$, an almost optimal harmonic broadcast protocol is introduced by Kuhn et al.~\cite{dualgraph} to solve the problem against 0-oblivious adversary. 
	 Furthermore, Ghaffari et al.~\cite{obliviousDG} show that single-message broadcast in case of $T=\infty$ can be solved efficiently against $\infty$-oblivious adversary, while the problem cannot be solved  faster than $\Omega(n/\log n)$ against 1-oblivious adversary. 
	 In~\cite{opodis15}, we show the correlation between the complexity of the broadcast problem and the stability and connectivity of dynamic network (value of $T$) and also the adaptiveness of the adversary (value of $\tau$). 
\\[.1cm]
\noindent
\textit{Local Broadcast.} In this problem, some subset of nodes $B\subseteq V$ are initially given a message. 
	 Assume that set $R$ is the set of nodes with at least one neighbor in $B$. 
	 The problem is solved when all the nodes in $R$ receive at least one message from some node in $B$. 
	 In~\cite{obliviousDG}, the authors show that solving the problem in case of $T=\infty$ and $\tau = \infty$ requires at least $\Omega(\sqrt{n}/\log n)$ rounds. 
	 In the same setting, but with some geographic constraint, it is proved that the problem can be solved in $O(\log^2 n\log \Delta)$ rounds, which is within a log-factor of the optimal solution in the static protocol model.
	 Ghaffari et al.~\cite{ghaffari12} compare the hardness of solving this problem in the classic static model and dynamic networks with $\infty$-interval connectivity. 
\\[.1cm]
\noindent
\textit{Multi-message Broadcast.} 
	Multi-message broadcast or $s$-token dissemination problem in static networks is well-studied (see \cite{attiya,leighton,lynch,peleg}), where the problem can be solved in time $O(n+s)$~\cite{topkis85}. 
	However, this problem in dynamic networks is first considered in~\cite{kuhn:stoc10}, where a $O(nk)$ algorithm is proposed for the case that nodes do not know $n$. 
	Moreover, a lower bound of $\Omega(n\log k)$ is given in the mentioned paper which is later improved by Dutta et al.~\cite{dutta} to an almost tight lower bound of $\Omega(nk/\log n+n)$. 
	Haeupler and Kuhn~\cite{haeupler:disc12} extended this lower bound for the cases where nodes are allowed to send $b\geq1$ tokens per round, the graph in each round is $c$-vertex connected, the general case of $T$-interval connectivity, and finally $\delta$ fraction of nodes are only required to receive each token. 
	Sarma et al.~\cite{sarma} proposed a random walk based algorithm for solving token dissemination in dynamic networks. 
	 Abshoff et al.~\cite{abshoff} introduce an almost matching lower and upper bounds for solving $s$-token dissemination in dynamic networks within unit disk graph model, where a worst-case adversary can move the nodes in a Euclidean plane with maximum velocity of at most $v_{max}$. 
	 Clementi et al.~\cite{clementi:podc12} study the flooding time in Markovian dynamic evolving dynamic graphs where the graph changes according to a random process. 
	  }
	

\section{Model and Problem Definition}
\label{sec:model}
\textbf{Dynamic Networks.}  We model dynamic radio networks using the
synchronous dynamic network model of \cite{kuhn:stoc10}.\footnote{We note
  that many other, similar dynamic network models have appeared in the
  literature (cf. Section~\ref{sec:intro}).} A dynamic network is
represented by a fixed set of nodes $V$ of size $n$ and a sequence of
undirected graphs $\langle G_1, G_2, \dots\rangle$, where
$G_i=(V,E_i)$ is the communication graph in round $i$. Hence, while
the set of nodes remains the same throughout an execution, the set of
edges can potentially change from round to round. A dynamic graph
$\langle G_1,G_2,\dots\rangle$ is called \textit{$T$-interval
  connected} for an integer parameter $T\geq 1$ if the graph
\[
\overline{G}_{r,T} = (V,\overline{E}_{r,T}), \quad\text{where }
\overline{E}_{r,T} := \bigcap_{r'=r}^{r+T-1} E_{r'}
\]
is connected for all $r\geq 1$.

\vspace{.1cm}
\noindent
\textbf{Communication Model.}  We define an $n$-node distributed
algorithm $\mathcal{A}$ as a collection of $n$ processes which are
assigned to the nodes of an $n$-node network. Thus, at the beginning
of an execution, a bijection from $V$ to $\mathcal{A}$ is defined by
an adversary. In the following, we will use ``node'' to refer to a
node $u\in V$ and to the process running at node $u$. We assume that
each node has a unique ID of $O(\log n)$ bits. In each round of an
execution, each node $u$ decides to either transmit a message or
listen to the wireless channel. When a node decides to transmit a
message in round $r$, the message reaches all of its neighbors in
$G_r$. A node $u$ in round $r$ successfully receives a message from a
neighbor $v$ if and only if $v$ is the only neighbor of $u$ in $G_r$
transmitting in round $r$. Otherwise, if zero or multiple messages
reach a node $u$, $u$ receives silence, i.e., nodes cannot
detect collisions.\\[.1cm]
\noindent
\textbf{Adversary.}
We assume that the dynamic graph $\langle G_1,G_2, \dots\rangle$ is
determined by an adversary. Classically, in this context, three types
of adversaries have been considered (see, e.g., \cite{obliviousDG}).
An \emph{oblivious} adversary has to determine the whole sequence of
graphs at the beginning of an execution, \emph{independently} of any
randomness used in the algorithm. An \emph{adaptive} adversary
can construct the graph of round $r$ depending on the history and thus
in particular the randomness up to
round $r$. Typically, two different adaptive adversaries are
considered. A \emph{strongly adaptive} adversary can choose graph
$G_r$ dependent on the history up to round $r$ including the
randomness of the algorithm in round $r$, whereas a \emph{weakly
  adaptive} adversary can only use the randomness up to round $r-1$ to
determine $G_r$. In the present paper, we use a more fine-grained
adversary definition which was in this form introduced in
\cite{opodis15}. For an integer $\tau \geq 0$, we call an adversary
\textit{$\tau$-oblivious} if for any $r\geq 1$, the adversary
constructs $G_r$ based on the knowledge of the algorithm description
and the algorithm's random choices of the first $r-\tau$ rounds. Note
that the three classic adversaries described above are all special
cases of a $\tau$-oblivious adversary, where $\tau=\infty$ corresponds
to an oblivious adversary, $\tau=0$ to a strongly adaptive adversary
and $\tau=1$ to a weakly adaptive adversary.\\[.1cm]
\noindent
\textbf{Multi-message Broadcast Problem.}  For some positive integers
$s$ and $B$, we define the $s$-multi-message broadcast problem as
follows.  We assume that there are $s$ distinct \textit{broadcast
  messages} of size $B$ bits, each of which is given to some node in
the network, called a \textit{source node}.  Then the problem requires
the dissemination of these $s$ broadcast messages to all nodes in the
network.  It can be solved by two types of algorithms;
\textit{store-and-forward algorithms} and \textit{network coding
  algorithms}. In store-and-forward algorithms, each broadcast message
is considered as a black box, and nodes can only store and forward
them.  In contrast to store-and-forward algorithms, in network coding
algorithms, each node can send a message which can be any function of
the messages it has received so far.
In the following we define two types of problems and later we show that a procedure that can solve each of these problems can be used as a subroutine for solving the multi-message broadcast problem.\\[.1cm]
\noindent
\textbf{Communication Capacity.}  To solve the multi-message broadcast
problem we consider a restriction on the amount of data which can be
transmitted in a single message of a distributed protocol. The
\textit{communication capacity} $c\geq 1$ is defined as the maximum
number of broadcast messages that can be sent in a single message of a
distributed algorithm. In addition, an algorithm can send some
additonal control information of the same (asymptotic) size, i.e., an
algorithm is allowed to send messages of $O(cB)$ bits.
\\[.1cm]
\noindent
\textbf{Limited Single-Message Broadcast Problem.} 
	Let us assume that there is a single broadcast message initially given to some node in an $n$-node network. 
	For some integer parameter $k\leq n$, $k$-limited broadcast
        problem requires the successful receipt of the broadcast
        message by at least $k$ nodes (including the source) in the network with probability at least $1/2$. \\[.1cm]
\noindent
\textbf{Concurrency-Resistant Single-Message
  Broadcast Problem.}
    In this problem we assume that there can be 0, 1, or more source nodes in the network, each of which is given a broadcast message. 
    If there exists only one source node, then its broadcast message is required to be successfully received by all nodes in the network, and thus this execution is considered a successful broadcast. 
    Otherwise, we have an unsuccessful broadcast, where there are no source nodes or more than one source nodes.
    The problem requires that all nodes detect whether the broadcast was successful or not by the end of the execution. 
    This broadcast with success detection is actually simulating a single-hop communication network with collision detection, where if only one node broadcasts in a round, all nodes receive the message and otherwise all nodes detect collision/silence.


\section{Multi-message Broadcast Algorithms}
\label{sec:algorithm}

In this section we present upper bounds for the multi-message
broadcast problem in different scenarios depending on the
communication
capacity, 
the interval connectivity of the communication network, the adversary
strength, and also the ability to use network coding for disseminating
information. We start by describing \emph{generic techniques} for
broadcasting multiple messages in dynamic radio networks. Further, the discussion of how
to use network coding to solve multi-message broadcast in dynamic
radio networks appears in Section~\ref{sec:apps}.

\subsection{Generic Algorithm for Large Communication Capacity 
\boldmath$c$}
\label{sec:generic1}
We start by describing a generic method for coping with a general
communication capacity parameter $c$ (see Section~\ref{sec:model}). If we
want to design store-and-forward algorithms which exploit the fact
that in a given time slot a node can transmit $c\gg 1$ source
messages to its neighbors, we have to deal with the problem that
initially, all broadcast messages might start at distinct source
nodes. In order to collect sufficiently many broadcast messages at
single nodes and to nevertheless avoid too many redundant
retransmission of the same broadcast message by several nodes, we
adapt a technique introduced by Chrobak et al.\ in
\cite{chrobak:2004}. Their algorithm runs in phases and it is based on
iterative applications of a $k$-limited single-message broadcast
routine for exponentially growing values of $k$.  In each phase, for
each broadcast message \calM, the minimum number of nodes which know
\calM doubles. Typically, the maximum time for reaching at least $k$
nodes with a single-message broadcast algorithm grows linearly in
$k$. If $k$ is doubled in each phase, the time for each $k$-limited
single-message broadcast instance therefore also doubles. However,
since in each phase, the number of nodes which know each source
message \calM at the beginning of a phase also doubles, the number of
source nodes from which we need to start a $k$-limited single-message
broadcast instance gets divided by $2$ and overall, the time
complexity of each phase will be about the same.

\vspace{.1cm}
\noindent
\textbf{Distributed Coupon Collection.} Formally, in
\cite{chrobak:2004} this idea is modeled by a \emph{distributed coupon
  collection} process which we generalize here. The
distributed coupon collection problem is characterized by four
positive integer parameters $n$, $s$, $\ell$, and $c$ (in
\cite{chrobak:2004}, the parameters $s$ and $c$ are both equal to
$n$). There are $n$ bins and $s$ distinct coupons (in the application,
the bins will correspond to the nodes of the communication network and
the distinct coupons will be the broadcast messages). There are at
least $\ell$ copies of each coupon and the at least $s\ell$ coupons
are distributed among the $n$ bins such that each bin contains at most
one copy of each of the $s$ distinct coupons. The coupon collection
proceeds in discrete time steps. In each time step, a bin is chosen
uniformly at random. If a bin is chosen, it is opened with probability
$1/2$. If the bin is opened, at most $c$ coupons of it are collected
as follows. If the bin has at most $c$ coupons, all coupons
of the bin are collected, otherwise a randomly chosen subset of size $c$ of
the coupons in the bin is collected. The coupon collection ends
as soon as each of the $s$ distinct coupons has been collected at
least once. The following lemma upper bounds the total number of time
steps needed to collect all $s$ coupons.

\begin{lemma}[Distributed Coupon Collection]\label{lemma:couponcoll}
  With high probability, the described distributed coupon collection
  process ends after $O\big(\frac{ns}{c\ell}\cdot \log(n+s)\big)$ steps.
\end{lemma}
\begin{proof}
	Consider a particular coupon $x$ and one step of the distributed 
coupon collection process. 
	Since there are at least $\ell$ copies of $x$, each in a different bin, there are at least $\ell$ bins containing $x$. 
	Thus, by picking a random bin, we pick a bin having $x$ with probability at least $\ell/n$. 
	Moreover, as there are $s$ distinct coupons and each bin can hold at most one copy of each coupon, when collecting a random subset of $c$ of the coupons in a
  bin containing $x$, the probability for collecting $x$ is at least
  $c/s$.
  
 	Overall, in each step of the process, the probability for collecting $x$ is therefore at least $\frac{c\ell}{ns}$. 
	To collect $x$ with probability at least $1-\delta$ for some $\delta>0$, we therefore need at most $\frac{ns}{c\ell}\cdot\ln(1/\delta)$ steps. 
	The lemma then follows by setting $\delta=\frac{1}{s n^\gamma}$ for a constant $\gamma>0$ and by applying a union bound over all the $s$ distinct coupons.
	
\end{proof}

\vspace{.1cm}
\noindent
\textbf{Coupon-Collection-Based Generic Multi-message Broadcast
  Algorithm.}
We now discuss how to use the above abstract distributed coupon
collection process to efficiently broadcast multiple messages in a
dynamic radio network. Our algorithm is a generalization of the idea
of Chrobak et al.~\cite{chrobak:2004}. The algorithm consists of
$\lceil\log_2 n\rceil$ phases. By the end of phase $i$, for each
broadcast message $\calM$, we want to
guarantee that at least $\ell_i := 2^i$ nodes know $\calM$. We can
therefore assume that at the beginning of a phase, each source
message $\calM$ is known by at least $\ell_{i-1}=2^{i-1}$ nodes (note
that this is trivially true for $i=1$). We can achieve that each
broadcast message is known by $\ell_i$ nodes by running sufficiently
many instances of $\ell_i$-limited broadcast such that each source
message $\calM$ is disseminated in $O(\log n)$ of these
$\ell_i$-limited broadcast instances. The details of the algorithm are
given by Algorithm~\ref{alg:CCLalgorithm}. In the pseudocode of
Algorithm~\ref{alg:CCLalgorithm}, $\alpha>0$ is a constant which is chosen
sufficiently large.
 
\begin{algorithm}[ht]
\caption{Generic Multi-message Broadcast Algorithm Based on
  Distributed Coupon Collection.}
\label{alg:CCLalgorithm}
\begin{algorithmic}[1]
  \State{\textbf{for each} $v\in V$ \textbf{do} $S_v\gets$ set of
    broadcast messages known by $v$}
  \For{$i\gets1$ \textbf{to} $\lceil\log_2 n\rceil$}
  \State $\ell_i\gets 2^i$
  \For{$j\gets1$ \textbf{to} $\alpha\cdot \frac{ns}{c\ell_i}\cdot \ln n$}
  \For {\textbf{each} $v\in V$} 
  \State (independently) mark $v$ with probability $1/n$
  \If{$v$ is marked}
  \State $R_v\gets$ random subset of $S_v$ of size $\min\set{|S_v|,
    c}$
  \State $v$ initiates $\ell_i$-limited broadcast with message $R_v$
  \EndIf
  \EndFor
  \EndFor
  \EndFor
\end{algorithmic}
\end{algorithm}

The following lemma shows that if the cost of $k$-limited broadcast
depends at most linearly on $k$, Algorithm~\ref{alg:CCLalgorithm} achieves
$s$-multi-message broadcast in essentially the time needed to perform
$s/c$ single-message broadcasts (i.e., $n$-limited broadcasts).

\begin{lemma}\label{lem:CCLalgorithm}
  Let $t(k,n)$ be the time needed to run one instance of $k$-limited
  broadcast in an $n$-node network. Then, when using a sufficiently
  large constant $\alpha>0$ in Algorithm~\ref{alg:CCLalgorithm}, the algorithm
  w.h.p.\ solves $s$-multi-message broadcast in time
  \[
  O\left(
  \sum_{i=1}^{\lceil\log_2 n\rceil} \frac{ns}{c2^i}\cdot 
  t\big(2^i,n)\cdot\log n
  \right).
  \]
\end{lemma}
\begin{proof}
  Recall that if only one node in the network
  initiates the given $k$-limited broadcast routine with broadcast
  message \calM, then with probability at least $1/2$, at least $k$
  nodes receive \calM successfully.  Fix some broadcast message
  $\zeta$ as one of the $s$ broadcast messages initially given to
  source nodes. We show by induction that, w.h.p., at the beginning of
  the $i^{\mathit{th}}$ iteration (and thus at the end of the
  $(i-1)^{\mathit{st}}$ iteration) of the outer for loop, there are at
  least $\ell_{i-1}$ nodes which know $\zeta$. Clearly, this is true
  for $i=1$. For $i\geq 1$, consider the $i^{\mathit{th}}$ iteration of the
  outer for loop of the algorithm.  Assuming that there exist at least
  $\ell_{i-1}$ nodes having $\zeta$, any node which has $\zeta$
  chooses $\zeta$ to pack it in its sent message with probability at
  least $c/s$ in each iteration of the inner for loop.  Furthermore,
  each node is the only marked node in the network in each iteration
  of the inner for loop with probability at least
  $\frac{1}{n}\big(1-\frac{1}{n}\big)^{n-1}\geq\frac{1}{en}$.
  Therefore, $\zeta$ will be packed into the sent message of a node which is
  the only marked node in that round with probability at least
  $\frac{c\ell_{i-1}}{ens}$ and in this case, it will reach at least
  $\ell_i$ nodes with probability at least $1/2$.  Thus, $\zeta$ will
  be known by at least $\ell_i$ nodes after
  $\frac{\alpha ns}{c\ell_i}\ln n = \frac{\alpha ns}{2c\ell_{i-1}}\ln
  n$
  repetitions of the inner for loop with probability at least
  $1-n^{-\alpha /4e}$.
  
\end{proof}

Typically, the time to run one instance of $k$-limited broadcast
depends linearly on $k$. The following corollary simplifies the
statement of Lemma~\ref{lem:CCLalgorithm} for this particular case.

\begin{corollary}\label{cor:CCLalgorithm}
  Assume that the time for running one instance of $k$-limited
  broadcast in an $n$-node network is given by $t(k,n)\leq k\cdot
  t(n)$. Then, when using a sufficiently
  large constant $\alpha>0$ in Algorithm~\ref{alg:CCLalgorithm}, the algorithm
  w.h.p.\ solves $s$-multi-message broadcast in time
  \[
  O\left(t(n) \cdot \frac{ns}{c}\cdot \log^2n\right).
  \]
\end{corollary}
\begin{proof}
  Follows directly from Lemma~\ref{lem:CCLalgorithm}.
\end{proof}

\subsection{Generic Algorithm for Constant Communication Capacity \boldmath$c$}
\label{sec:generic2}
We next describe a more efficient generic algorithm for the case $c=1$
(or any constant $c$). Hence, we assume that each message sent by the
algorithm can only contain a single broadcast message.  In this
case, we do not need to care about collecting different broadcast
messages at a single node and Algorithm~\ref{alg:CCLalgorithm} is therefore too
costly.  Instead, we use an algorithm which is based on a more standard
single-message broadcast algorithm.  In the following, we assume that
for some given setting we are given a \emph{concurrency-resistant}
single-message broadcast algorithm $\mathcal{B}$.  Recall that if there exists
only one broadcast message while running $\mathcal{B}$, all nodes
receive the broadcast message and $\mathcal{B}$ returns 1, otherwise
it returns 0.  We use this algorithm as a subroutine in designing
the generic multi-message broadcast algorithm. We assume that
initially, the number of broadcast messages $s$ is known. The
generic algorithm runs in phases and in each phase we run one instance
of $\mathcal{B}$. Note that if the instance returns $1$, all nodes
know the broadcast message which has been disseminated to all
nodes. Therefore at all times, all nodes know how many messages still
need to be broadcast. If at the beginning of a phase, there are
$x\leq s$ broadcast messages which still need to be broadcast, for
each broadcast message $\mathcal{M}$, the source node of $\mathcal{M}$ decides to start an instance of $\mathcal{B}$
with broadcast message $\mathcal{M}$ with probability $1/x$.

\begin{algorithm}[ht]
\caption{Generic Multi-Message Broadcast Algorithm Based on
  Concurrency-Resistant Single-Message Broadcast Algorithm.}
\label{alg:Galgorithm}
\begin{algorithmic}[1]
  \State{\textbf{for each} $v\in V$ \textbf{do} $S_v\gets$ set of
    broadcast messages given initially to $v$}
  \State{$x \gets s$}  
  \While{$x\neq 0$}
  \State{\textbf{for each} $v\in V$ \textbf{do}}
  \State{\hspace{.4cm} \textbf{if} $S_v\neq \emptyset$ \textbf{then}}
  \State{\hspace{.7cm} (independently) mark each of the broadcast messages in $S_v$ with probability $1/x$}
  \State{\hspace{.4cm} \textbf{if} $v$ has a marked broadcast message \textbf{then}}
  \State{\hspace{.7cm} $v$ initiates $\mathcal{B}$ with its marked broadcast messages}
  \State{\hspace{.4cm} \textbf{if} $\mathcal{B}$ returns 1 \textbf{then}}
  \State{\hspace{.7cm} $x\gets x-1$}
  \State{\hspace{.7cm} remove the broadcast message of node $u$ which
    is delivered to all nodes from $S_u$}
  \State{\hspace{.4cm} unmark all broadcast messages}

  \EndWhile
\end{algorithmic}
\end{algorithm}

	
\begin{lemma}\label{lem:Galgorithm}
  If $\mathcal{B}$ is a concurrency-resistant single-message broadcast
  algorithm with running time $t(n)$, then
  Algorithm~\ref{alg:Galgorithm} solves the $s$-multi-message broadcast problem
  in time $O\big((s+\log n)t(n)\big)$.
\end{lemma}
\begin{proof}
  In each phase for some $1\leq x\leq s$, there are $x$ broadcast
  messages left to be received by all.  Each of these broadcast
  messages is marked to be broadcasted with probability $1/x$.
  Therefore, the probability of each broadcast message to be the only
  marked broadcast message in the phase is $p$, given by
  \[
  p = \frac{1}{x}\left(1-\frac{1}{x}\right)^{(x-1)} \geq \frac{1}{ex}.
  \]
  Let us call a phase which has exactly one marked broadcast message a \textit{successful} phase. 
  Therefore, the probability that a phase with $x$ broadcast messages is a successful phase is at least $1/e$. 
  Hence if the number of phases is at least $c'\log n$ for some large constant $c'$, we have $s$ successful phases with high probability, which concludes the proof.  
\end{proof}

In the sequel, we apply this generic broadcast algorithm to solve the
multi-message broadcast problem in three different settings.

\subsection{Application of the Generic Algorithms in Different
  Settings}

In this section we intend to show how to apply the generic
multi-message broadcast techniques introduced in Section~\ref{sec:generic1}
and Section~\ref{sec:generic2} in different settings.  In the following we
consider the problem in networks with different interval connectivity
$T$. In each of these settings we investigate the two cases of large
and constant communication capacities separately.  For the first case
the generic algorithm requires the existence of a $k$-limited
single-message broadcast algorithm and in the second case the generic
algorithm requires the existence of a concurrency-resistant
single-message broadcast algorithm. We therefore need to show that the
existing single-message broadcast algorithms in the considered dynamic
radio network settings can be turned into $k$-limited and
concurrency-resistant variants of these algorithms.

\subsubsection{(Setting I) \boldmath$\infty$-Interval Connected Dynamic Networks}

We consider the $s$-multi-message broadcast problem in an
$\infty$-interval connected dynamic network against a $0$-oblivious
adversary. To obtain a $k$-limited and a concurrency-resistant
single-message broadcast algorithm, we adapt the harmonic broadcast
algorithm introduced in \cite{dualgraph}. In the harmonic broadcast
algorithm, in the first round, the source node transmits its broadcast
message to its neighbors.  From the second round on, any node which
receives the broadcast message in round $r_v$, transmits the message in any round
$r>r_v$ with probability $p_v(r)$, given by
\[
		p_v(r) := \frac{1}{1+\lfloor \frac{r-r_v-1}{\mathcal{T}}   \rfloor},
\]
where $\mathcal{T}=\Theta(\log n)$ is a parameter. That is, for the
first $\mathcal{T}$ rounds immediately after some node receives the
message, it transmits the message with probability 1, for the next
$\mathcal{T}$ rounds it transmits with probability $1/2$, then with
probability $1/3$ and so on. 

\vspace*{.1cm}
\noindent
\textit{Procedure Harmonic  $k$-Limited Broadcast.}\\
	By adapting the harmonic broadcast algorithm in~\cite{dualgraph}  we limit the number of nodes that receive the message successfully to $k$. 
	For this purpose we only run the harmonic broadcast algorithm for $4k\mathcal{T}(\ln (n)+1)$ rounds, where $\mathcal{T} = \lceil 12\ln (n/\epsilon) \rceil$ and $\epsilon$ is a small value which will be fixed later.  
	
	We start with some definitions. For any $r\geq 1$ we define
        $P(r)$ as the sum of transmitting probabilities of all nodes
        in round $r$. Any round $r$ with $P(r)<1$ is called
        \textit{free} and any other round is called a \textit{busy} round. If some node $v$ is the only node that is transmitting in some round $r$, then we say that node $v$ is \textit{isolated} in round $r$. 
\begin{lemma}(Lemma 13 in~\cite{dualgraph})
\label{lem:kuhn}
	Consider a node $v$ and let $r_v$ be the round when $v$ first
        receives the broadcast message. 
	Further, let $r>r_v$ be such that at least half of the rounds $r_v+1, \dots , r$ are free. 
	If $\mathcal{T}\geq 12\ln (n/\epsilon)$ for some $\epsilon>0$, then with probability larger than $1-\epsilon/n$ there exists a round $r'\in[r_v+1,r]$ such that $v$ is isolated in round $r'$. 
\end{lemma}
\begin{lemma}\label{lem:lim_bc}
  Let us assume that a source node, by transmitting message \calM,
  initiates the Harmonic $k$-Limited Broadcast procedure in an
  $\infty$-interval connected dynamic network controlled by a
  0-oblivious adversary. Then with probability larger than
  $1-\epsilon$, by the end of the procedure at least $\min\set{k,n}$ nodes
  receive $\mathcal{M}$ successfully.
\end{lemma}

\begin{proof}
  We prove this lemma based on the analysis in~\cite{opodis15}.  For
  the sake of contradiction let us assume that only $k'<\min\set{k,n}$ nodes
  receive $\mathcal{M}$ by time
  $t := 4k\mathcal{T}\left(\ln (n)+1\right)$.  Therefore, the sum of
  all transmitting probabilities of these $k'$ nodes until time $t$ is
  at most
\[
		k' \mathcal{T} \sum\limits_{i=1}^{\lceil\frac{t}{\mathcal{T}}\rceil}  \frac{1}{i}
		\leq 
		k \mathcal{T} \left(\ln \left(\left\lceil\frac{t}{\mathcal{T}}\right\rceil \right) + 1\right)
		\leq 
	\frac{t-1}{2}.
\]
We can conclude that the number of free rounds is more than the number
of busy rounds until time $t$.  Since in the first round one source
node initiates the execution by transmitting with probability 1, the
first round is busy.  We define $\theta_0:=0$, and for $i>0$,
$\theta_i$ is the first time after time $\theta_{i-1}$ that the number
of free and busy rounds in the time interval
$[\theta_{i-1}, \theta_{i}]$ are equal.  Moreover, let $m$ be a
positive integer such that $\theta_m$ is the last such time in the
time interval $[0,t]$.  For each $i\in\set{1,\dots,m}$, let us assume
that the number of nodes that get informed until time $\theta_i$ is
$k_i$ (note that $k_i\leq k'$ for all $i\leq m$). In order to complete
the proof of the lemma, we next proof three helper claims.

\smallskip

\noindent
\begin{center}
  \begin{minipage}{0.96\linewidth}
    \textbf{Claim 1.} \label{claim31} If a node $v$ gets informed in round
    $t_v$, where $\theta_{i-1} \leq t_v < \theta_i$, and round
    $\theta_{i-1}+1$ is busy, then $v$ gets isolated by time $\theta_i$
    with high probability.
    \begin{proof}[Proof of Claim 1]
      Let $\hat{t}$ be the first time after $t_v$ such that the number of
      free and busy rounds in $[t_v,\hat{t}]$ are equal. We show that if
      round $\theta_{i-1}+1$ is busy, it holds that
      $\hat{t}\leq \theta_i$. The claim that we need to proof then
      directly follows from Lemma~\ref{lem:kuhn}. For the sake of
      contradiction, let us assume that $\hat{t}>\theta_i$. Note that node
      $v$ transmits with probability $1$ in round $t_v+1$ and therefore
      round $t_v+1$ is busy. Assuming that $\hat{t}>\theta_i$ therefore
      implies that the number of free rounds is less than the number of
      busy rounds in $[t_v,\theta_i]$. However, from the minimality of
      $\theta_{i}$ and the assumption that round $\theta_{i-1}+1$ is busy,
      we can also conclude that the number of free rounds is less than the
      number of busy rounds in $[\theta_{i-1},t_v]$.  Hence, in the whole
      interval $[\theta_{i-1},\theta_i]$, there are more busy than free
      rounds, which contradicts the definition of $\theta_i$ and thus our
      assumption that $\hat{t}>\theta_i$, which concludes the proof of the
      claim.
    \end{proof}
  \end{minipage}
\end{center}

\smallskip

\noindent
\begin{center}
  \begin{minipage}{0.96\linewidth}
    \textbf{Claim 2.} \label{claim32} 
    For all $i\in\set{1,\dots,m}$, round $\theta_i+1$ is busy and
    the $k_i^{\mathit{th}}$ node that gets informed, gets informed in round $\theta_i$.
    \begin{proof}[Proof of Claim 2]
      We prove the claim by induction on $i$. First note that round
      $1=\theta_0+1$ is busy because in the first round, the source
      node transmits with probability $1$. We next show that for every
      $i\in \set{1,\dots,m}$, if round $\theta_{i-1}+1$ is busy, the
      $k_i^{\mathit{th}}$ node is informed in round $\theta_i$. For the sake of
      contradiction, assume that the $k_i^{\mathit{th}}$ node is informed
      before round $\theta_i$. Because round $\theta_{i-1}+1$ is
      assumed to be busy, Claim 1 then implies that the first $k_I$
      informed nodes all get isolated by round $\theta_i$. Because for
      all $i\leq m$, we have $k_i\leq k'<\min\set{k,n}$, there is a
      at least one node which does not get informed in the first
      $\theta_i$ rounds. Hence, in the stable connected backbone
      graph, there is at least one node $v$ which is not informed in
      the first $\theta_i$ rounds and which is connected to one of the
      first $k_i$ informed nodes. When this node gets isolated (by
      time $\theta_i$, node
      $v$ gets informed, a contradiction to the assumption that all
      $k_I$ nodes are informed before round $\theta_i$. It remains to
      prove that round $\theta_i+1$ is busy. We have already seen that
      this is true for $i=0$. For $i\geq 1$, it follows because the
      $k_i^{\mathit{th}}$ informed node gets informed in round $\theta_i$ and
      it therefore transmits with probability $1$ in round $\theta_i+1$.
    \end{proof}
  \end{minipage}
\end{center}

\smallskip

\noindent
From Claim 2, it follows that round $\theta_m+1$ is busy. Because time
$\theta_m$ is the latest time smaller than or equal to $t$ for such
that the number of busy rounds equals the number of free rounds, this
implies that in the first $t$ rounds, there are more busy than free
rounds. However, we have shown that from assuming that
$k'<\min\set{k,n}$, it follows that in the first $t$ rounds, we have
more free than busy rounds. We thus get a contradiction to the
assumption that the number of informed nodes is less than $\min\set{k,n}$.
\end{proof}

The following lemma shows the existence of a concurrency-resistant single-message broadcast algorithm to be used in the second generic algorithm. 

\begin{lemma}\label{lem:resistance1}
		There exists a concurrency-resistant algorithm which solves the single-message broadcast problem in $\infty$-interval connected dynamic $n$-node networks against a 0-oblivious adversary in $O(n\log^2n)$ rounds. 
\end{lemma}
\begin{proof}
	Consider an algorithm which runs in two phases.
	Each phase is a single run of the harmonic broadcast algorithm in~\cite{dualgraph}.
	In the first phase all the source nodes initiate the broadcast algorithm.
	In the second phase, each node that has received at least two messages in the first phase initiates the broadcast algorithm with the broadcast message $\bot$. 
	At the end of the second phase, each node that receives only one broadcast message in the first phase and nothing in the second phase, detects a successful broadcast.
	Otherwise, each node that receives nothing in the first phase or receives $\bot$ in the second phase, detects an unsuccessful broadcast.
	
	If there exists no source nodes, then in the first phase all nodes receive nothing and therefore they detect the unsuccessful broadcast correctly. 
	Hence, for this case even the first phase is enough for all the nodes to detect the failure of the broadcast.
	Now, let us consider the case when there are more than one source nodes. 
	Based on the analysis of the algorithm in \cite{dualgraph}, due to the transmitting probability of the nodes in the network all the nodes will be informed by a broadcast message until the end of the first phase. 
	Since we have a connected network there exists at least two nodes which are connected in the stable backbone of the network and have received different broadcast messages.
	Hence, after their isolation at least one of them knows at least two different broadcast messages. 
	Then in the second phase of the execution, there exists at least one node that initiates the broadcast with $\bot$, and hence all the nodes receive $\bot$ by the end of the second phase.
\end{proof}

\noindent
\textbf{Theorem~\ref{thm:01} (restated).} Assume that we are given an $\infty$-interval connected dynamic
  $n$-node network controlled by a $0$-oblivious adversary.  Using
  store-and-forward algorithms, for communication capacity $c=1$,
  $s$-multi-message broadcast can be solved, w.h.p., in time
  $O\big(ns\log^2n\big)$, whereas for arbitrary $c\geq 1$, it can be
  solved, w.h.p., in time $O\big(\big(1+\frac{s}{c}\big)n\log^4 n\big)$.
\begin{proof}
	Let us assume Algorithm~\ref{alg:CCLalgorithm} performs the $k$-limited broadcasting by running the Harmonic $k$-Limited Broadcast procedure. 
	Considering the selection of $\alpha$ in Algorithm~\ref{alg:CCLalgorithm}, it is enough to set $\epsilon$ in the Harmonic $k$-Limited Broadcast procedure to $1/2$. 
	Then it takes $t(n,k) \leq a k \ln^2n$ rounds for a sufficiently large constant $a$. 
	Therefore, $t(n,k) = k t(n)$, where $t(n) = O(\log^2n)$. 
	Due to Corollary~\ref{cor:CCLalgorithm}, Algorithm~\ref{alg:CCLalgorithm} solves the $s$-multi-message broadcast problem in this setting in time $O\left(\left(1+\frac{s}{c}\right)n\log^4 n\right)$. 
	
	On the other hand, Lemma~\ref{lem:resistance1} shows the existence of a concurrency-resistant single-message broadcast algorithm with running time of $O(n\log^2n)$.
	Hence, by embedding the concurrency-resistant single-message broadcast algorithm in Algorithm~\ref{alg:Galgorithm}, the problem in this setting, based on Lemma~\ref{lem:Galgorithm}, can be solved in time 
\[
	O\left((s+\log n)n\log^2n\right).
\]
\end{proof}

\subsubsection{(Setting II) \boldmath$1$-Interval Connected
  Dynamic Networks}

For $1$-interval connected networks, we assume that the adversary is
$1$-oblivious. Note that in \cite{opodis15}, it is shown that even the
single-message problem cannot be solved with a $0$-oblivious adversary
in $1$-interval connected networks. We adapt the homogenous algorithm by
Clementi et al.~\cite{clementi09} to obtain a $k$-limited and a
concurrency-resistant single-message broadcast procedure. The
single-message broadcast algorithm of \cite{clementi09} is a simple
randomized algorithm where every node that knows the broadcast
message, broadcasts the message with probability $\ln(n)/n$. In
\cite{clementi09}, it is shown that in every round, the probability
that a node knowing the broadcast message succeeds in sending it to a
neighboring node which does not know the message is at least
$\ln(n)/n$. In order to have a $k$-limited broadcast algorithm which
succeeds in reaching at least $k$ nodes with probability at least
$1/2$, we therefore need to run the algorithm of \cite{clementi09} for
$\Theta(kn/\ln n)$ rounds. Hence, the second statement of Theorem~\ref{thm:02}
follows directly from Corollary~\ref{cor:CCLalgorithm}.

In order to obtain a concurrency-resistant algorithm, we can first
observe that the algorithm of \cite{clementi09} also works if several
nodes try to broadcast the same message. Hence, in order to use a
similar idea as in the $\infty$-interval connected case, we have to
make sure that in the case of multiple broadcast messages, some node
detects that there are at least two messages. Using the same analysis
as for a single broadcast message, one can see that after
$O(n^2/\log n)$ rounds, every node knows at least one of the broadcast
messages. Our next goal is to have at least one node which knows at
least two different broadcast messages. As long as this is not the
case, in every round, there is at least one
edge connecting two nodes $u$ and $v$ which know different broadcast messages. For
the next $O(n\log n)$ rounds, each node now transmits the broadcast
message it knows with probability $1/n$ (assuming a node knows exactly
one broadcast message). The probability that either $u$ or $v$ decide
to broadcast their message and no other node in the network decides to
broadcast is at least $2/en$ and thus after $O(n\log n)$ rounds,
w.h.p., there is at least one node which knows two broadcast messages.

Now, we can get a concurrency-resistant single-message broadcast
algorithm in the same way as for $\infty$-interval connected
graphs. We run the algorithm of \cite{clementi09} again such that each
node which knows at least two different broadcast messages after the
first phase, broadcasts $\bot$. The time complexity of the
concurrency-resistant broadcast algorithm is $O(n^2/\log n)$ and
therefore the first statement of Theorem~\ref{thm:02} now directly follows from
Lemma~\ref{lem:Galgorithm}.

\subsubsection{(Setting III) \boldmath$T$-Interval Connected
  Dynamic Networks}

For arbitrary $T$-interval connected graphs, we adapt the algorithm of
\cite{opodis15}. We assume that the adversary is $\tau$-oblivious and
we define $\psi:=\min\set{\tau,T}$. We assume that
$\psi = \Omega(\log^2 n)$, note that otherwise, Theorem~\ref{thm:03} already
follows from the statement of Theorem~\ref{thm:02} for $T=1$. We also assume that
$\psi=O(n)$, as otherwise Theorem~\ref{thm:03} follows from Theorem~\ref{thm:01} for
$T=\infty$. The algorithm of \cite{opodis15} runs in phases of length
$\psi$. The progress in a single phase is analyzed in the proof of
Theorem~\ref{thm:01}. A phase of length $\psi$ is called successful if at least
an $\Omega(\psi/(n\log^2n))$-fraction of all uninformed nodes are
informed and it is shown, in the proof of Theorem 1 of~\cite{opodis15}, that a phase is successful with probability
at least $\psi/(8en)$.

Assume that we want to obtain a $k$-limited single-message broadcast
algorithm. If $k\leq n/2$, as long as fewer than $k$ nodes are
informed, the number of uninformed nodes is at least $n/2$. Hence, in
this case, in a successful phase, at least $\Omega(\psi/\log^2n)$
nodes are newly informed. Hence, to inform $k\leq n/2$ nodes, we need
$O(k\log^2(n)/\psi)$ successful phases and we thus need
$O(kn\log^2(n)/\psi^2)$ phases in total to inform $k$ nodes with
probability at least $1/2$. The number of rounds to solve $k$-limited
single-message broadcast is therefore
$O(kn\log^2(n)/\psi)=O(kn\log^2(n)/\min\set{\tau,T})$. For $k>n/2$, we
run the complete single-message broadcast algorithm of \cite{opodis15}
and obtain a time complexity of
$O(n^2\log^3(n)/\min\set{\tau,T})$. The second statement of
Theorem~\ref{thm:03} now follows directly from Lemma~\ref{lem:CCLalgorithm}.

To obtain a concurrency-resistant algorithm we use the same approach
as in the $1$-interval connected case. We first run the algorithm of
\cite{opodis15}. From the description and analysis of the algorithm is
not hard to see that the algorithm also works if several source nodes
start the algorithm with the same broadcast message. In addition, the
transmission behavior of a node does not depend on the content of the
broadcast message and thus, if there are $2$ or more broadcast
messages, one can easily see that at the end of the algorithm, either
there exists at least one node which knows at least two different broadcast messages or each node knows exactly one message. If every node
knows exactly one broadcast message, we can create one node which
knows at least two different broadcast messages in $O(n\log n)$ rounds
in the same way as in the $1$-interval connected case. The
concurrency-resistant algorithm is then completed by running the
algorithm of \cite{opodis15} once more, where every node which knows
at least two different broadcast messages, uses the algorithm to
broadcast $\bot$ to all nodes. The total time complexity of the
algorithm is $O(n^2\log^3(n)/\min\set{\tau,T} + n\log n)$ and
therefore the first statement of Theorem~\ref{thm:03} follows from
Lemma~\ref{lem:Galgorithm}.

\section{Multi-message Broadcast Using Network Coding}
\label{sec:apps}
	To solve the multi-message broadcast problem in this setting we use randomized linear network coding (RLNC) to increase throughput while using an adapted version of the homogeneous randomized broadcast protocol by Clementi et al.~\cite{clementi09} as the underlying broadcast protocol. 
	The RLNC algorithm thus tells the nodes \textit{what} to send and the broadcast protocol tells them \textit{when} to send. 
	We describe and analyze the RLNC algorithm based on the elegant work by Haeupler~\cite{haeupler11}. 
	
	Let us assume that $\mathcal{M}_1, \mathcal{M}_2, \dots , \mathcal{M}_s$ are the $s$ broadcast messages initially given to the source nodes. 
	We can represent these messages as vectors over a finite field $\mathbb{F}_q$, where $q$ is a prime or prime power.
	For all $1\leq i \leq s$, let $\overrightarrow{m}_i \in
        \mathbb{F}_q^l$ be the vector representation for
        $\mathcal{M}_i$, where $l$ is the maximum length of a message in
        base $q$ notation.
	In each round, each node that decides to transmit, sends a packet $(\overrightarrow{\mu}, \overrightarrow{m})$, where $\overrightarrow{\mu}$ is a coefficient vector and $\overrightarrow{m}$ is a message vector (i.e., a linear combination of the broadcast messages) such that
\begin{align*}
		\overrightarrow{\mu}  = (\mu_1, \mu_2, \dots , \mu_s) \in \mathbb{F}_q^s \ \ \ \ \ \text{and} \ \ \ \ \ \ 
		\overrightarrow{m}  = \sum\limits_{i=1}^{s}\mu_i\overrightarrow{m}_i \in \mathbb{F}_q^l.
\end{align*}
	When a node has received $s$ packets with $s$ linearly independent coefficient vectors, it can reconstruct all the $s$ broadcast messages by applying Gaussian elimination. 
	In other words, when the span of the received coefficient vectors  by a node is the full space $\mathbb{F}_q^s$, then it can decode all the broadcast messages. 
	Here we explain the algorithm to solve the multi-message broadcast problem in this setting. 
	
\noindent
\textit{Algorithm.}\\
	Initially, each node has only one packet. 
	A source node which has broadcast message $\mathcal{M}_i$, has the packet $(\overrightarrow{\mu}, \overrightarrow{m}_i)$, where $\overrightarrow{\mu}$ is the $i^{\mathit{th}}$ standard basis vector.
	A node which is not a source node has the packet $(\overrightarrow{0}, \overrightarrow{0})$. 
	During the execution, each node $v$ keeps a complete span $\psi(v)$ of all the packets it has received (i.e., $\psi(v)$ is the set of all the linear combinations of the received packets). 
	In each round, each node $v$ chooses a packet, uniformly at random, from $\psi(v)$ and transmits it with probability $1/n$. 
\\[.1cm]
\noindent
\textit{Analysis.}\\
	Let $\chi(v)$ denote the set of all linear combinations of the coefficient vectors in $\psi(v)$ for node $v$. 
	During the algorithm execution, for any node $v$, by receiving more messages the coefficient subspace $\chi(v)$ grows monotonically to the full space $\mathbb{F}_q^s$. 
	We say that node $v$ in round $r$ \textit{knows about} $\overrightarrow{\mu} \in \mathbb{F}_q^s$, if $\overrightarrow{\mu}$ is not orthogonal to $\chi(v)$ in that round, that is, there exits some vector $\overrightarrow{c} \in \chi(v)$ such that $\langle \overrightarrow{c}, \overrightarrow{\mu} \rangle \neq 0$. 

\begin{lemma}(Lemma 4.2 from~\cite{haeupler11})\label{lem:haeupler}
		If node $u$ knows about vector $\overrightarrow{\mu} \in \mathbb{F}_q^s$ and transmits a packet to node $v$, then $v$ knows about $\overrightarrow{\mu}$ afterwards with probability at least $1-1/q$. 
		Furthermore, if a node knows about all the vectors in $\mathbb{F}_q^s$, it is able to decode all the $s$ broadcast messages.
\end{lemma}	
\begin{lemma}
\label{lem:vectorbc}
	In each round of the algorithm execution, for any vector $\overrightarrow{\mu}\in \mathbb{F}_q^s$, if there exists some node that does not know about $\overrightarrow{\mu}$, then there exists at least one new node knowing about $\overrightarrow{\mu}$ afterwards with probability at least $1/2en$. 
\end{lemma}
\begin{proof}
		Let us fix an arbitrary vector $\overrightarrow{\mu} \in \mathbb{F}_q^s$ and some positive integer $r$ such that there exists some node in the network in round $r$ that does not know about $\overrightarrow{\mu}$. 
		Since the network is connected, there exists some node $v$ which does not know about $\overrightarrow{\mu}$ and is connected to node $u$ which knows about $\overrightarrow{\mu}$, in round $r$. 
		knowing the fact that all the informed nodes transmit with probability $1/n$, the probability that $u$ is the only neighbor of $v$ transmitting in round $r$ is at least 
\[
		\frac{1}{n}\left(1- \frac{1}{n}\right)^{n-1} > \frac{1}{en}.
\]
		Based on Lemma~\ref{lem:haeupler}, if $u$ is the only transmitting neighbor of $v$, then $v$ knows about $\overrightarrow{\mu}$ afterwards with probability $1-1/q\geq 1/2$.
		Therefore, in each such a round, one new node knows about $\overrightarrow{\mu}$ afterwards with probability at least $1/2en$.
		
\end{proof}
\noindent
\textbf{Theorem~\ref{thm:04} (restated).}
	Using linear network coding, in $1$-interval connected dynamic networks with communication capacity $1$ and a $1$-oblivious adversary, $s$-multi-message broadcast can be solved in time $O(n^2+ ns)$, w.h.p.
\begin{proof}
		Let us assume that the algorithm runs for $cn(n+s)$ rounds for some constant $c$. 
		Then due to Lemma~\ref{lem:vectorbc}, for any arbitrary coefficient vector $\overrightarrow{\mu} \in \mathbb{F}_q^s$, the number of nodes that know about it by the end of the algorithm execution is dominated by a binomial random variable $X\sim Bin(cn(n+s), 1/2en)$.
		 Choosing a large enough constant as $c$ and applying Chernoff bound, we can show that the probability of having one node not knowing about $\overrightarrow{\mu}$ is less than $1/nq^s$.
		 Applying union bound over all $q^s$ coefficient vectors in $\mathbb{F}_q^s$, all nodes know about all coefficient vectors in $\mathbb{F}_q^s$ w.h.p. , hence, they can decode all the broadcast messages. 		
\end{proof}

\hide{	
\subsection{Multi-message broadcast in {\boldmath$T$}-interval connected dynamic networks}
	To solve the multi-message broadcast problem in $T$-interval connected networks against a $\tau$-oblivious adversary, we generalize the randomized algorithm introduced by Ahmadi et al.~\cite{opodis15} which solves the single message broadcast in the similar setting. 
	The algorithm has a time complexity of $O\big((1+n/\psi)\cdot ns\log^3n\big)$, where $\psi := min\{\tau , \mathcal{T}\}$.
	Let us assume that there are $s$ source nodes with broadcast messages $\mathcal{M}_1, \mathcal{M}_2, \dots , \mathcal{M}_s$.
	The algorithm runs in phases that are defined as follows.
{\color{red} Should we here cite~\cite{kuhn:podc10} and~\cite{clementi09} or it's fine?}

\begin{definition}[Phase]\label{de:phase}
  The $j^{\mathit{th}}$ time interval of $\tphase$ consecutive rounds is called
  phase $j$, where $j$ is a positive integer. Hence, phase $j$ starts
  at time $(j-1)\tphase$ and ends at time $j\tphase$ and it consists of
  rounds $(j-1)\tphase+1,\dots,j\tphase$.
\end{definition} 
	For each node $v$, let $t_{v,i}$ be the round where $v$ receives $\mathcal{M}_i$ for the first time.
	Fix an arbitrary round $t$ in phase $j$. 
	For each node $v$ in round $t$, the set of messages are divided into three subsets; The set $\mathcal{A}$ of messages that $v$ received for the first time in previous phases (the phases $1$ to $j-1$), the set $\mathcal{B}$ of messages that $v$ received for the first time in phase $j$ until round $t$, the set $\mathcal{C}$ of messages that $v$ has not received yet. Accordingly, we say that $v$ is \textit{previously informed} about set $\mathcal{A}$, \textit{newly informed} about set $\mathcal{B}$, and \textit{uninformed} about set $\mathcal{C}$.  

	The algorithm behaves differently in the two halves of each phase during the execution.
	In the first half of each phase, any node that is previously or newly informed about a non-empty set, chooses one of the messages it has received uniformly at random and then transmits it with probability $1/n$. 
	In the second half of the phase, only the nodes that are newly informed about a non-empty set of broadcast messages participate in the execution. In each round $t$, any participant $v$ decides to transmit each message $\mathcal{M}_i$ with \textit{decision probability} $p_{v,i}$, given by
\begin{equation}\label{eq:harmonicalg}
  \forall t > \floor{ \frac{t}{\tphase}} \cdot \tphase
  +\ceil{\frac{\tphase}{2}}:\, \ \  \forall v \in V : \ \ \ \ \ \ 
  p_{v,i}(t) :=
  \begin{cases}
  \frac{1}{1+ \floor{\frac{t-\hat{t}_{v,i}-1}{\mathcal{T}}}},
  \ \ \ & \text{ if } \floor{t/\tphase} \cdot \tphase < t_{v,i} < t\\
  0 & \text{ otherwise, } 
  \end{cases}
\end{equation}
where 
\[
\hat{t}_{v,i} := \begin{cases}
  \floor{ \frac{t}{\tphase}} \cdot \tphase
  +\ceil{\frac{\tphase}{2}}, & \text{ if } \floor{ t/\tphase} \cdot
  \tphase < t_{v,i} < \floor{t/\tphase} \cdot \tphase +\ceil{\tphase/2}\\
  t_{v,i}, & \text{ if } \floor{t/\tphase} \cdot \tphase +\ceil{\tphase/2} < t_{v,i} < t
\end{cases}
\]
and $\mathcal{T}$ will be fixed later. For any node $v\in V$ and any $i \in [s]$, if $\mathcal{M}_i$ is the only message that $v$ decides to transmit, $v$ transmits $\mathcal{M}_i$ in round $t$. 
	Otherwise, it remains silent.

	Thus, in the second half of the phase, for any message $\mathcal{M}_i$, $1\leq i \leq s$, any node that receives it for the first time in the first half of the current phase, starts considering its transmission in the first round of the second half and if the node receives it for the first time in the second half, it starts considering transmission of the message immediately from the next round. 
	When the node starts considering transmission of a message, for the first $\mathcal{T}$ rounds the decision probability is 1, for the next $\mathcal{T}$ rounds is $1/2$, then $1/3$ and so on. \\[.3cm]
\noindent
\textbf{Analysis.} 
	The analysis is based on the analysis in~\cite{opodis15} but a more involved one. Let $P(t)$ denotes the sum of decision probabilities of all messages by all nodes in round $t$, i.e., $P(t) := \sum\limits_{v\in V} \sum\limits_{i=1}^s p_{v,i}(t)$. Since the graph is $T$-interval connected, during each phase $j$, we have a stable spanning subgraph which is called the \textit{backbone} of the phase, denoted by $\Psi_j$. 
	If $P(t)\geq 1$, we call $t$ a \textit{busy} round and otherwise a \textit{free} round. If in round $t$, node $v$ is the only node that transmits a message, we say that $v$ is \textit{isolated} on that message in round $t$. 
	In an execution of the algorithm, for any message $\mathcal{M}_i$ we consider a non-decreasing sequence $\mathcal{W}_i = t_{1,i}, t_{2,i}, \dots, t_{s,i}$, where $t_{1,i} = 0$ and $t_{j,i}$ is the round in which the $j^{\mathit{th}}$ node receives $\mathcal{M}_i$ for the first time. We call $\mathcal{W}_i$ the wake-up pattern of message $\mathcal{M}_i$ in the execution. 
	We can consequently consider the set $\mathcal{W} = \{\mathcal{W}_1, \mathcal{W}_2, \dots, \mathcal{W}_s\}$, as the complete wake up pattern of the execution.
	Notice that the transmitting probabilities of the nodes during the execution can be determined by the complete wake up pattern. 
	With the same argumentation in~\cite{kuhn:podc10} we can adapt Lemma 10 and 11 to have the following lemma. 
\begin{lemma} (adapted from \cite{kuhn:podc10})
For any possible set of wake-up patterns $\{\mathcal{W}_1, \mathcal{W}_2, \dots, \mathcal{W}_s\}$, the total number of busy rounds is at most $ns \mathcal{T} H(ns)$. 
\end{lemma} 
\noindent
Similar to the proof of Lemma 13 in~\cite{kuhn:podc10} we can prove the following adapted lemma. 
\begin{lemma}\label{lem:iso} (adapted from \cite{kuhn:podc10})
Consider a node $v$, and let $t_{v,i}$ be the time when $v$ receives message $\mathcal{M}_i$ for the first time. Further, let $t>t_{v,i}$ be such that at least half of the rounds $t_{v,i}+1, \dots , t$ are free. If $\mathcal{T} \geq 12\ln (ns/\epsilon)$ for some $\epsilon > 0$, then with probability larger than $1-\epsilon /ns $ there exists a round $t' \in [t_{v,i}+1,t]$ such that $v$ is isolated on $\mathcal{M}_i$ in round $t'$. 
\end{lemma}

\noindent
We define $\theta_0:=0$, and for $\ell>0$, $\theta_\ell$ is the first time after time $\theta_{\ell-1}$ that the number of free and busy rounds in the time interval $[\theta_{\ell-1}, \theta_{\ell}]$ are equal. 
	Moreover, let $m$ be a non-negative integer such that $\theta_m$ is the last such time in the time interval $[0,t]$. 

\begin{lemma}
	For all phases, in each time interval $[\theta_{\ell-1}, \theta_\ell]$, where $\ell\in [s]$, if round $\theta_{\ell-1}+1$ is busy then any node $v$ with $\hat{t}_{v,i}\in \{\theta_{\ell-1},\dots , \theta_\ell-1\}$ gets isolated on $\mathcal{M}_i$ in some round $t'\in \{ \hat{t}_{v,i}+1, \dots , \theta_\ell\}$ with high probability. 
\end{lemma}

\begin{definition}
	Node $v$ is called ``available to know $\mathcal{M}_i$'' in round $r$ of phase $j$, if either it is newly informed about $\mathcal{M}_i$ or there exists a path in $\Psi_i$ from $v$ to a node which is newly informed about $\mathcal{M}_i$. 
	Furthermore, the number of nodes in round $r$ that are available to know $\mathcal{M}_i$ is called the number of ``availabilities for message $\mathcal{M}_i$''. 
	Moreover, the event that a node is newly informed about some message is called one ``learning'' for that message. 
\end{definition}

\begin{lemma}
	Fix an arbitrary phase and a positive integer $\ell$. If at the beginning of each phase there exists at least one node which is uninformed about $\mathcal{M}_i$ but available to know it, then w.h.p. at least one node gets informed about $\mathcal{M}_i$ in round $\theta_\ell$.
\end{lemma}
\begin{proof}

\end{proof}

\begin{lemma}
	Consider an arbitrary phase and assume that at the beginning of the second $\lfloor \psi/2\rfloor$ rounds of the phase totally there are $z$ availabilities for the broadcast messages. Then, w.h.p., for some constant $c$ the total number of learnings in this phase is at least $\min\{z, c\psi /\ln^2n\}$.
\end{lemma}
\begin{proof}

\end{proof}

\noindent
\textbf{Theorem 1.4. (restated)} Let $T \geq 1$ and $\tau \geq 1$ be positive integer parameters. Multi-message broadcast with $s$ broadcast messages in $T$-interval connected dynamic networks against a $\tau$-oblivious adversary can be solved in time
\[
 O\left(\left(1+\frac{n}{min\{\tau , T\}}\right)\cdot ns\log^3n\right),
 \]
 where communication capacity is 1. 
 
\begin{proof}
 	
\end{proof}

\vspace{1cm}
{\color{red} incomplete!}
}



\section{Multi-message Broadcast Lower Bound}
\label{sec:lower}
	
In this section we give a lower bound for solving the multi-message
broadcast problem in an $\infty$-interval connected network controlled
by a $0$-oblivious adversary.

Recall that in an $\infty$-interval connected dynamic network, there
is a static connected spanning subgraph which is present
every round. In the following, we refer to this graph as the stable
subgraph. We first prove a simple lower bound for the case where the
stable subgraph has a non-constant diameter. With a more involved
argument, we then extend the lower bound to the case where the
stable subgraph has a constant diameter.

\begin{lemma}\label{lemma:simplelower}
  Assume that $G$ is an $n$-node graph with maximum degree
  $\Delta=O(1)$. If $G$ is the stable subgraph of an $\infty$-interval
  connected dynamic radio network, with a $0$-oblivious adversary,
  every $s$-multi-message broadcast algorithms requires at least
  $\Omega(ns/c)$ rounds, where $c\geq 1$ is the communication capacity
  of the network.
\end{lemma}
\begin{proof}
  Recall that a $0$-oblivious adversary can construct the
  communication graph of a given round $r$ after the random decisions
  of all nodes in round $r$. A $0$-oblivious adversary therefore in
  particular knows which nodes are transmitting in round $r$ before
  determining the graph of round $r$. Given an $s$-multi-message
  broadcast algorithm $\mathcal{A}$, a $0$-oblivious adversary
  constructs the sequence of communication graphs as follows. In every
  round in which $2$ or more nodes decide to transmit, the
  communication graph is a complete graph. In all other rounds, the
  communication graph is only the stable graph $G$. Hence, in rounds
  with $2$ or more nodes transmitting, all $n$ nodes will experience a
  collision and therefore the algorithm cannot make any progress. In
  rounds where $0$ nodes are transmitting, there clearly also cannot
  be any progress. In rounds where exactly one node $v$ is
  transmitting, the message of $v$ only reaches its at most
  $\Delta=O(1)$ neighbors in $G$. Because we have $s$ broadcast
  messages of $B$ bits and because each broadcast message only has one
  source node, over the whole algorithm, the nodes in total need to
  learn $\Theta(nsB)$ bits of information. As every message sent by
  the algorithm can contain only $O(cB)$, in each round, the total
  number of bits learned by any node is also at most $O(cB)$. The
  lemma therefore follows.
\end{proof}

To prove the lower bound for constant-diameter stable subgraphs, we
use the hitting game technique introduced by Newport
in~\cite{newport:disc14}.  This is a general technique to prove lower
bounds for solving various communication problems in radio networks.
Using this technique, one first defines an appropriate combinatorial
game with respect to the problem such that a lower bound for the game
can be proved directly.  It is shown that an efficient solution
for the radio network problem helps a player to win the game
efficiently.  Consequently, the game's lower bound can be leveraged to
the problem's lower bound.
	
For the sake of proving this lower bound, we define a combinatorial
game called $(\alpha , \beta)$-hitting game.  Assuming the existence
of a distributed algorithm $\mathcal{A}$ which solves multi-message
broadcast in the desired setting, we will show that a player can
simulate the execution of $\mathcal{A}$ in an $\infty$-interval
connected dynamic network called the \textit{target network}. Then,
the player uses the transmitting behavior of the nodes in the target
network, while running $\mathcal{A}$, to win the game efficiently. 

\vspace{.1cm}
\noindent
\textbf{\boldmath$(\alpha , \beta)$-hitting Game.} The game is defined for two positive integers $\alpha$ and $\beta$, where $\beta < \alpha$. It proceeds in rounds, and a player, which is represented by a probabilistic automaton $\mathcal{P}$, plays the game against a referee. 
	At the beginning of the game, the referee arbitrarily partitions the set $[\alpha +\beta]$
	\footnote{For two integers $a\leq b$, $[a,b]$ denotes the set of all integers between $a$ and $b$ (including $a$ and $b$). Further, for an integer $a\geq 1$, we use $[a]$ as a short form to denote $[a]:=[1,a]$.} into two disjoint sets $A$ and $B$, such that $|A| = \alpha$ and $|B| = \beta$. 
	This partitioning can be seen by the player and its soul purpose is the ease of target network construction that will be explained later. 
	Accordingly, the referee selects $\beta$ elements from $A$ uniformly at random, and it generates a random permutation of these elements represented by $\langle a_1, a_2, \dots , a_\beta \rangle$. 
	The \textit{target set} $R$ is defined based on this sequence, given by
\[
R = \{ (a_1,1), (a_2,2), \dots , (a_\beta , \beta) \},
\]
	is kept secret from the player. 
	In each round, the player proposes a guess $(x,y)$, and the referee responses to the player whether the guess was in the target set or not. 
	In case the guess belongs to $R$, the referee removes the guess from $R$ at the end of the round. 
	The player wins the game after $r$ rounds, if and only if at the end of round $r$ the set $R$ is empty. 
	The following lemma shows a lower bound for this game, which is adapted from Lemma 3.2 in~\cite{obliviousDG}.

\begin{lemma}\label{lem:game}
	(Adapted from~\cite{obliviousDG}) There does not exist a player that can win the $(\alpha , \beta)$-hitting game in $o(\alpha \beta)$ rounds with high probability
	\footnote{We say that a probability event happens with high probability (w.h.p.) if it happens with probability at least $1-1/n^c$, where $n$ is the number of nodes and $c>0$ is a constant which can be chosen arbitrarily large by adjusting other constants.}.
\end{lemma}

\begin{lemma}\label{lem:simulation}
	Let $s$ and $n$ be positive integers, where $s<n$. If algorithm $\mathcal{A}$ solves multi-message broadcast problem with $s$ source nodes in any $\infty$-interval connected dynamic $n$-node network against a $0$-oblivious adversary in $f(n,s)$ rounds in expectation, then it is possible to win the $(n-s, s)$-hitting game in expected $O(f(n,s))$ rounds.  
\end{lemma}

\begin{proof}
  For the following discussion in this proof, we fix an arbitrary
  instance $\mathcal{I}$ of the $(\alpha , \beta)$-hitting game, where
  $\alpha = n-s$ and $\beta = s$.  Based on $\mathcal{I}$, we first
  define a particular $\infty$-interval connected dynamic $n$-node
  network called \textit{target network}.  To guarantee the
  $\infty$-interval connectivity of the target network we show that
  there exits a fixed connected $n$-node graph we call the
  \textit{$s$-clique-star} which is a subgraph of the graph
  representations of the target network in all rounds of any 
  execution.  Since one cannot construct the target network based on
  an instance of the game with no knowledge about the secret
  information that the referee has, the player simulates the execution
  of $\mathcal{A}$ on the target network.  However, we will see later
  that with the information revealed gradually by the referee during
  the game and the power of a $0$-oblivious adversary, the simulation by
  the player is completely consistent with the execution of
  $\mathcal{A}$ on the target network.  Then we show that how the
  simulation allows the player to win the game by the end of the
  broadcast.  We start with the definition of the $s$-clique-star
  graph and the target network.

\noindent
\textbf{\boldmath$s$-Clique-Star Graph.} 
	Consider an $n$-node static graph $G$ which is defined for some positive integer $s < n$. 
	The nodes are partitioned into two disjoint sets of size $s$ and $n-s$, such that the $n-s$ nodes form a clique, and each of the $s$ nodes is connected to exactly one node in the clique. 
	Each of these $s$ nodes is called an \textit{external node}, and any node in the clique which is connected to at least one external node is called a \textit{bridge node}. Rest of the nodes are called \textit{internal nodes}.
You can see the graph in Figure~\ref{fig:star}.

\begin{figure}[h]
	\centering
\begin{tikzpicture} 	
	\draw (0,0) circle (1.5);
	\foreach \x in {30, 90,
						    165, 210, 
						   270, 330}
	{
	\draw[fill]  (\x:1.2cm) circle (.05);
	}

	\foreach \x in {0, 30, 90,
						    165, 210, 
						   270, 330}
	{
	\draw[fill]  (\x:2.3cm) circle (.05);
	}
	
	\foreach \x in {30, 90,
						    165, 210, 
						   270, 330}
	{
	\draw  (\x:1.2cm) -- (\x:2.3cm);
	}
	\draw  (0:2.3cm) -- (330:1.2cm);
	
	\foreach \x in {125, 135, 145}
	{
	\draw[fill]  (\x:1.7cm) circle (.03);
	}
	
	\node at (30:2.6cm) {$e_1$};
	\node at (0:2.6cm) {$e_2$};
	\node at (330:2.6cm) {$e_3$};
	\node at (270:2.6cm) {$e_4$};
	\node at (210:2.6cm) {$e_5$};
	\node at (165:2.6cm) {$e_6$};
	\node at (90:2.6cm) {$e_s$};
	
	\node at (30:.9cm) {$b_1$};
	\node at (330:.9cm) {$b_2$};
	\node at (270:.9cm) {$b_3$};
	\node at (210:.9cm) {$b_4$};
	\node at (165:.9cm) {$b_5$};
	\node at (90:.9cm) {$b_j$};
\end{tikzpicture}
\caption{The $s$-clique-star graph. For some positive integer $j\leq s$, $b_1, b_2, \dots , b_j$ are bridge nodes, and $e_1, e_2, \dots , e_s$ are external nodes.}
\label{fig:star}
\end{figure}
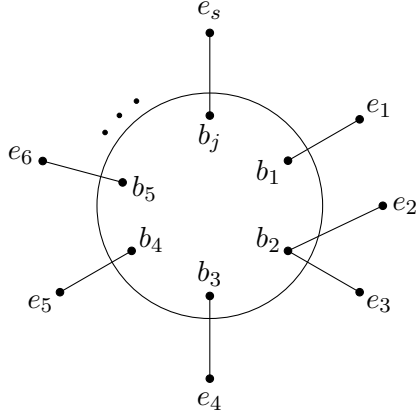

\noindent
\textbf{The Target Network.}
	The target network is an $\infty$-interval connected dynamic $n$-node network $G^t = \langle G_1, G_2, \dots \rangle$, such that for all $r\geq 1$, $G_r$ is either a complete graph or a complete graph lacking exactly one edge.
	Let us assume that in the game instance $\mathcal{I}$, the partitioned sets determined by the referee are $A$ and $B$, and the target set is
\[
R = \big\{(a_1,1), (a_2, 2), \dots , (a_s, s)\big\},
\]
	such that among $a_1, \dots , a_s$, there are $j$ different values, where $j\leq s$. 
	Based on this instance of the game we first construct and instantiate
	\footnote{By instantiation of an $n$-node network,
we mean assigning $n$ processes with unique IDs to the nodes of the network.}
	 a fixed $n$-node $s$-clique-star graph $G^{cs}$ with $j$ bridge nodes. 
	 Then for any $i\geq 1$, we construct $G_i$ by adding edges to $G^{cs}$. 
	 
	 We construct and instantiate $G^{cs}$ as follows. We have $n$ processes with IDs from $1$ to $n$.
	The processes with IDs equal to the mentioned $j$ different values are assigned to the $j$ bridge nodes of $G^{cs}$. 
	Additionally, the $s$ processes whose IDs belong to $B$ are assigned to the external nodes. 
	We connect each external node $e_i$, where $1\leq i \leq s$, to the bridge node with ID equal to $a_i$. 
	Moreover, we randomly assign the broadcast messages
        $\mathcal{M}_1:=1, \mathcal{M}_2:=2, \dots , \mathcal{M}_s:=s$
        to processes assigned to the clique.

	\hide{In addition, we make sure that this external node has
          all the messages except the broadcast message $\mathcal{M}_i$ (this state of the target network can easily be achieved by the player with the power of a $0$-oblivious adversary, without violating $\infty$-interval connectivity property, which will be explained later).}

	In any round $r$, for constructing the target network we consider the following three possibilities of transmitting behavior of the processes in $G^t$: $(1)$ more than one node transmit, $(2)$ only one node transmits and it is a bridge node transmitting $\mathcal{M}_i$ such that node $e_i$ is its neighbor in $G^{cs}$, $(3)$ otherwise.
	In case of possibility (1) or (2), $G_r$ is a complete graph over all $n$ nodes. 
	In case of possibility (3), $G_r$ contains all edges among the nodes except the edge between the transmitter which is transmitting $\mathcal{M}_i$ and the node $e_i$. 

	For a complete construction and instantiation of the target network based on an instance of the game, one needs to also have the secret information that the referee of the game has (since it sometimes needs to recognize the bridge nodes). However, we show that without knowing the secret information the player can simulate execution of $\mathcal{A}$ on the target network with the gradual information that the referee reveals during the game. 

\vspace{.1cm}
\noindent
\textbf{The Simulation.}
	The player simulates the execution of $\mathcal{A}$ on the target network round by round. 
	In each round, according to the transmitting behavior of the nodes it generates at most one guess for the game. 
	Then playing the game for one round with the generated guess (if any), the player finishes the simulation of the current round by simulating the receive behavior of the nodes. 
	It continues simulating next rounds and playing the game until it wins the game. 
	In each round that only one node transmits, if it is not an external node, the player generates a guess consisting of its ID and the message it transmits. 			Otherwise, it does not generate any guess.
	 
	The receive behavior of the nodes in each round $r$ is simulated by the player based on the transmitting behavior of the nodes and the response of the referee to the generated guess.
	If more than one node transmits in a round, the player simulates all the nodes to receive collision. 
	If an external node is the only transmitter in round $r$ and it is transmitting message $\mathcal{M}_i$, then all the nodes except $e_i$ receive $\mathcal{M}_i$ successfully. 
	If a non-external node transmitting $\mathcal{M}_i$ is the only transmitter in round $r$, then the player generates a guess and plays the game. 
	If the guess is correct, then all the nodes receive $\mathcal{M}_i$. 
	Otherwise, all the nodes except $e_i$ receives $\mathcal{M}_i$ successfully. 
	
	The receive behavior of the nodes in the simulation is completely consistent with that of the execution of $\mathcal{A}$ on the target network. 
	In both of them whenever more than one node transmits all the nodes receive collision. 
	The only case in some round of the execution of $\mathcal{A}$ the target network is not a complete graph is when that a bridge node which has neighbor $e_i$ in $G^{cs}$ and transmitting $\mathcal{M}_i$. 
	And this case can be recognized by the player when it plays the game with the corresponding generated guess. 
	Therefore, we can conclude that the simulation by the player is completely consistent with the execution of $\mathcal{A}$ on the target network. 
	
	By the end of the simulation, the player wins the game.
	Considering the receive behavior of the nodes in the
        simulation, the only way that an external node receives the broadcast message $\mathcal{M}_i$ is to receive it directly from its corresponding bridge node in $G^{cs}$. 
	And, it has to happen in a round that there exists only one transmitter since otherwise all nodes receive collision.
	Moreover, by the end of the simulation, for all $i$, the external node $e_i$ should have received the message $\mathcal{M}_i$.  
	Hence, for all $i$, there should exists a round in the simulation that a bridge node transmitting $\mathcal{M}_i$ which is a neighbor of $e_i$ in $G^{cs}$ is the only transmitter. 
	Therefore, based on the guess generation approach of the player, we can conclude that the player guesses all the elements of the target set by the end of the broadcasting algorithm simulation.  
	
	Finally since in each round of the simulation the player generated at most one guess, if $\mathcal{A}$ solves the multi-message broadcast problem with $s<n$ source nodes in any $\infty$-interval connected network in $f(n,s)$ rounds, a player can win any instance of $(n-s, s)$-hitting game in at most $f(n,s)$ rounds by simulating the execution of $\mathcal{A}$ on the target network which is constructed based on the given instance of the game. 
\end{proof}

\noindent
\textbf{Theorem~\ref{thm:05} (restated).}
  In $\infty$-interval connected dynamic networks with communication
  capacity $c\geq 1$ and a $0$-oblivious adversary, any
  $s$-multi-message broadcast algorithm requires at least time
  $\Omega(ns/c)$, even when using network coding.  Further, there is a
  constant-diameter $\infty$-interval connected network with
  communication capacity $1$ such that any store-and-forward algorithm
  requires at least $\Omega\big((ns-s^2)/c\big)$ rounds to solve $s$-multi-message
  broadcast against a $0$-oblivious adversary.

\begin{proof}
	For the sake of contradiction let us assume that there exists an algorithm $\mathcal{A}$ that can solve the multi-message broadcast for $s$ source nodes in any $\infty$-interval connected dynamic $n$-node network in $f(n,s) = o(ns-s^2)$. Then based on Lemma~\ref{lem:simulation}, we can construct a player to win any instance $(\alpha, \beta)$-hitting game in $o(\alpha \beta)$ rounds. This contradicts Lemma~\ref{lem:game}. 
\end{proof}


\section{Conclusions}
In the paper, we studied multi-message broadcast in $T$-interval
connected radio networks with different adaptive adversaries. In all
considered cases, we shows that if $c$ broadcast messages can be
packed into a single message (of the algorithm), $s$ broadcast
messages can essentially be broadcast in $s/c$ times the time required
to broadcast a single message. In one case ($\infty$-interval
connected dynamic networks with a $0$-oblivious adversary), we also
showed that up to logarithmic factors, our algorithm is optimal. Note
that using techniques from \cite{obliviousDG,opodis15}, at the cost of
one logarithmic factor, this lower bound can also be adapted to work
in the presence of a $1$-oblivious adversary.

A multi-message broadcast time which is roughly $s/c$ times as large
as the time for broadcasting a single message seems not very
spectacular. Such an algorithm essentially always runs just one
single-message broadcast algorithm at each point in time (where for
$c>1$, the algorithm each time broadcasts a collection of
messages). However, we believe that it will be interesting to see
whether the time complexity can be significantly improved in any of 
the adversarial dynamic network settings considered in this
paper. When using store-and-forward algorithms, such an improvement
would imply that the algorithm can use some form of pipelining in an
efficient manner. I might also be interesting to study somewhat weaker
(adversarial) dynamic network models which allow some pipelining when
broadcasting multiple messages.



\vspace{1cm}
\bibliographystyle{abbrv}
\bibliography{references}

\begin{thebibliography}{10}

\bibitem{opodis15}
M.~Ahmadi, A.~Ghodselahi, F.~Kuhn, and A.~R. Molla.
\newblock The cost of global broadcast in dynamic radio networks.
\newblock In {\em Proc. of the 19th Int.\ Conf.\ on Principles of Distributed
  Systems (OPODIS)}, 2015.

\bibitem{fernandezanta12}
A.~F. Anta, A.~Milani, M.~A. Mosteiro, and S.~Zaks.
\newblock Opportunistic information dissemination in mobile ad-hoc networks:
  the profit of global synchrony.
\newblock {\em Distributed Computing}, 2012.

\bibitem{avin08}
C.~Avin, M.~Kouck\'{y}, and Z.~Lotker.
\newblock How to explore a fast-changing world (cover time of a simple random
  walk on evolving graphs).
\newblock In {\em Proc. of the 5th Coll. on Automata, Languages and Programming
  (ICALP)}, 2008.

\bibitem{bgi2}
R.~Bar-Yehuda, O.~Goldreich, and A.~Itai.
\newblock Efficient emulation of single-hop radio network with collision
  detection on multi-hop radio network with no collision detection.
\newblock {\em Distributed Computing}, 1991.

\bibitem{bgi1}
R.~Bar-Yehuda, O.~Goldreich, and A.~Itai.
\newblock On the time-complexity of broadcast in multi-hop radio networks: An
  exponential gap between determinism and randomization.
\newblock {\em Computer and System Sciences}, 1992.

\bibitem{baumann09}
H.~Baumann, P.~Crescenzi, and P.~Fraigniaud.
\newblock Parsimonious flooding in dynamic graphs.
\newblock In {\em Proc. of the 28th ACM Symp. on Principles of Distributed
  Computing (PODC)}, 2009.

\bibitem{time-varying}
A.~Casteigts, P.~Flocchini, W.~Quattrociocchi, and N.~Santoro.
\newblock Time-varying graphs and dynamic networks.
\newblock In {\em Proc. of the 10th Int. Conf. on Ad-hoc, Mobile, and Wireless
  Networks}, 2011.

\bibitem{DG_structuring}
K.~Censor-Hillel, S.~Gilbert, F.~Kuhn, N.~Lynch, and C.~Newport.
\newblock Structuring unreliable radio networks.
\newblock {\em Distributed Computing}, 2014.

\bibitem{chlamtac85}
I.~Chlamtac and S.~Kutten.
\newblock On broadcasting in radio networks--problem analysis and protocol
  design.
\newblock {\em IEEE Transactions on Communications}, 1985.

\bibitem{chlebus:2000}
B.~S. Chlebus, L.~Gasieniec, A.~Ostlin, and J.~M. Robson.
\newblock Deterministic radio broadcasting.
\newblock In {\em Proc. of the 27th Int. Coll. on Automata, Languages and
  Programming (ICALP)}, 2000.

\bibitem{Chlebus09}
B.~S. Chlebus, D.~R. Kowalski, and T.~Radzik.
\newblock Many-to-many communication in radio networks.
\newblock {\em Algorithmica}, 2009.

\bibitem{Christersson}
M.~Christersson, L.~Gasieniec, and A.~Lingas.
\newblock Gossiping with bounded size messages in ad hoc radio networks.
\newblock In {\em Proc. of the 29th Int. Coll. on Automata, Languages and
  Programming (ICALP)}, 2002.

\bibitem{chrobak:2004}
M.~Chrobak, L.~Gasieniec, and W.~Rytter.
\newblock A randomized algorithm for gossiping in radio networks.
\newblock {\em Networks}, 2004.

\bibitem{clementi09}
A.~Clementi, A.~Monti, F.~Pasquale, and R.~Silvestri.
\newblock Broadcasting in dynamic radio networks.
\newblock {\em Computer and System Sciences}.

\bibitem{clementi:2001}
A.~Clementi, A.~Monti, and R.~Silvestri.
\newblock Selective families, superimposed codes, and broadcasting on unknown
  radio networks.
\newblock In {\em Proc. of the ACM-SIAM Symp. on Discrete Algorithms (SODA)},
  2001.

\bibitem{clementi04}
A.~Clementi, A.~Monti, and R.~Silvestri.
\newblock Round robin is optimal for fault-tolerant broadcasting on wireless
  networks.
\newblock {\em Parallel \& Distributed Computing}, 2004.

\bibitem{gossip}
S.~Deb, M.~M{\'e}dard, and C.~Choute.
\newblock Algebraic gossip: A network coding approach to optimal multiple rumor
  mongering.
\newblock {\em IEEE/ACM Transactions on Information Theory}, 2006.

\bibitem{gasieniec}
L.~Gasieniec, E.~Kranakis, A.~Pelc, and Q.~Xin.
\newblock Deterministic m2m multicast in radio networks.
\newblock {\em Theoretical Computer Science}, 2006.

\bibitem{maclayer}
M.~Ghaffari, E.~Kantor, N.~Lynch, and C.~Newport.
\newblock Multi-message broadcast with abstract mac layers and unreliable
  links.
\newblock In {\em Proc. of the 2014 ACM Symp. on Principles of Distributed
  Computing (PODC)}, 2014.

\bibitem{obliviousDG}
M.~Ghaffari, N.~Lynch, and C.~Newport.
\newblock The cost of radio network broadcast for different models of
  unreliable links.
\newblock In {\em Proc. of the 32nd Symp.\ on Principles of Distributed
  Computing (PODC)}, 2013.

\bibitem{ghaffari16}
M.~Ghaffari and C.~Newport.
\newblock Leader election in unreliable radio networks.
\newblock In {\em Proc. of the 43rd Int.\ Coll.\ on Automata, Languages and
  Programming (ICALP)}, 2016.

\bibitem{gupta:2000}
P.~Gupta and P.~R. Kumar.
\newblock The capacity of wireless networks.
\newblock {\em IEEE Transactions on Information Theory}, 2000.

\bibitem{haeupler11}
B.~Haeupler.
\newblock Analyzing network coding gossip made easy.
\newblock In {\em Proc.of the 43rd ACM Symp. on Theory of Computing}, 2011.

\bibitem{khabbazian11}
M.~Khabbazian and D.~R. Kowalski.
\newblock Time-efficient randomized multiple-message broadcast in radio
  networks.
\newblock In {\em Proc. of the 30th ACM SIGACT-SIGOPS Symp. on Principles of
  Distributed Computing (PODC)}.

\bibitem{kim06}
K.-H. Kim and K.~G. Shin.
\newblock On accurate measurement of link quality in multi-hop wireless mesh
  networks.
\newblock In {\em Proc. of the 12th Int. Conf. on Mobile Computing and
  Networking (MOBICOM)}, 2006.

\bibitem{affectance}
D.~R. Kowalski, M.~A. Mosteiro, and T.~Rouse.
\newblock Dynamic multiple-message broadcast: Bounding throughput in the
  affectance model.
\newblock In {\em Proc. of the 10th ACM Int. Workshop on Foundations of Mobile
  Computing}, 2014.

\bibitem{kranakis01}
E.~Kranakis, D.~Krizanc, and A.~Pelc.
\newblock Fault-tolerant broadcasting in radio networks.
\newblock {\em Algorithms}, 2001.

\bibitem{kuhn:2011}
F.~Kuhn, N.~Lynch, and C.~Newport.
\newblock The abstract mac layer.
\newblock {\em Distributed Computing}, 2011.

\bibitem{dualgraph}
F.~Kuhn, N.~Lynch, C.~Newport, R.~Oshman, and A.~W. Richa.
\newblock Broadcasting in unreliable radio networks.
\newblock In {\em Proc. of the 29th Symp.\ on Principles of Distributed
  Computing (PODC)}, 2010.

\bibitem{kuhn:stoc10}
F.~Kuhn, N.~Lynch, and R.~Oshman.
\newblock Distributed computation in dynamic networks.
\newblock In {\em Proc. of the 42nd Symp.\ on Theory of Computing (STOC)},
  2010.

\bibitem{kuhn:2011:survey}
F.~Kuhn and R.~Oshman.
\newblock {Dynamic Networks: Models and Algorithms}.
\newblock {\em ACM SIGACT News}, 2011.

\bibitem{DG_localbroadcast}
N.~Lynch and C.~Newport.
\newblock A (truly) local broadcast layer for unreliable radio networks.
\newblock In {\em Proc.of the 34th Symp.\ on Principles of Distributed
  Computing (PODC)}, 2015.

\bibitem{moscibroda06}
T.~Moscibroda and R.~Wattenhofer.
\newblock The complexity of connectivity in wireless networks.
\newblock In {\em Proc.\ of the 25th Conf.\ on Computer Communications
  (INFOCOM)}, 2006.

\bibitem{newport:disc14}
C.~Newport.
\newblock Radio network lower bounds made easy.
\newblock In {\em Proc. of the 28th Int. Symp. on Distributed Computing
  (DISC)}. 2014.

\bibitem{newport:2007}
C.~Newport, D.~Kotz, Y.~Yuan, R.~S. Gray, J.~Liu, and C.~Elliott.
\newblock Experimental evaluation of wireless simulation assumptions.
\newblock {\em Simulation}, 2007.

\bibitem{odell05}
R.~O'Dell and R.~Wattenhofer.
\newblock Information dissemination in highly dynamic graphs.
\newblock In {\em Proc. of Workshop on Foundations of Mobile Computing
  (DIALM-POMC)}, 2005.

\bibitem{ramachandran07}
K.~Ramachandran, I.~Sheriff, E.~Belding, and K.~Almeroth.
\newblock Routing stability in static wireless mesh networks.
\newblock In {\em Proc. of the 8th Int. Conf.\ on Passive and Active Network
  Measurement}, 2007.

\bibitem{srinivasan08}
K.~Srinivasan, M.~A. Kazandjieva, S.~Agarwal, and P.~Levis.
\newblock The $\beta$-factor: Measuring wireless link burstiness.
\newblock In {\em Proc. of the 6th Conf. on Embedded Networked Sensor System},
  2008.

\bibitem{yarvis02}
M.~D. Yarvis, S.~W. Conner, L.~Krishnamurthy, J.~Chhabra, B.~Elliott, and
  A.~Mainwaring.
\newblock Real-world experiences with an interactive ad hoc sensor network.
\newblock In {\em Proc. of the Conf.\ on Parallel Processing}, 2002.

\end{thebibliography}

\end{document}